\documentclass[11pt]{article}

\usepackage{amsmath}
\usepackage{amsmath,amsfonts}
\usepackage{amsthm}
\usepackage{amssymb}
\usepackage{enumerate}
\usepackage{tweaklist}

\usepackage{tikz}
\usetikzlibrary{arrows,shapes,snakes,automata,backgrounds,petri}
\usepackage{hyperref}

\usepackage{epsfig, verbatim}
\usepackage{fullpage}
\usepackage{graphicx}
\usepackage{graphicx} 
\usepackage[T1]{fontenc}
\usepackage[utf8x]{inputenc}
\usepackage{csquotes}
\usepackage[english]{babel}
\usepackage{eurosym}
\usepackage{tikz}
\usepackage{verbatim}
\usepackage{fancybox}

\usepackage[dvips, paper=letterpaper, top=1in, bottom=.75in, left=1in, right=1in, nohead, includefoot, footskip=.25in]{geometry}

\usepackage{algorithmic}
\usepackage{algorithm}
\floatname{algorithm}{Procedure}

\renewcommand{\vec}{\textbf}

\usepackage{color}
\usepackage{subfig}
\usepackage{enumerate}
\usepackage{epsfig}
\usepackage[active]{srcltx}
\usepackage{amsfonts}
\usepackage{amsmath}

\newenvironment{prevproof}[2]{\noindent {\bf {Proof of {#1}~\ref{#2}:}}}{$\blacksquare$\vskip \belowdisplayskip}

\newtheorem{theorem}{Theorem}
\newtheorem{informal}{Informal Theorem}
\newtheorem{lemma}[theorem]{Lemma}

\newtheorem{fact}[theorem]{Fact}
\newtheorem{claim}{Claim}

\newcommand{\argmax}{\text{argmax}}
\newcommand{\poly}{\text{poly}}
\newcommand{\Opt}{\text{OPT}}

\newcommand{\C}{\mathcal{C}}
\newcommand{\notshow}[1]{{}}
\newcommand{\sout}[1]{{}}

\begin{document}
\title{On the Economic Efficiency of the Combinatorial Clock Auction}
{\author{Nicolas Bousquet\\
Department of Mathematics and Statistics, McGill University \\
and GERAD, Université de Montréal\\
\tt{nicolas.bousquet2@mail.mcgill.ca}
\and
Yang Cai\\
School of Computer Science, McGill University\\
\tt{cai@cs.mcgill.ca}
\and Christoph Hunkenschr\"oder\\
University of Bonn\\
\tt{chr.hunkenschroeder@gmail.com}
\and
Adrian Vetta\\
Department of Mathematics and Statistics,
and School of Computer Science\\
McGill University\\
\tt{vetta@math.mcgill.ca}
}}
\maketitle
\date{}
\begin{abstract}
Since the 1990s spectrum auctions have been implemented world-wide.
This has provided for a practical examination of an assortment of auction mechanisms and, amongst these,
two simultaneous ascending price auctions have proved to be extremely successful.
These are the simultaneous multiround ascending auction (SMRA) and the combinatorial
clock auction (CCA). It has long been known that, for certain classes of valuation functions,
the SMRA provides good theoretical guarantees on social welfare.
However, no such guarantees were known for the CCA.

In this paper, we show that CCA does provide strong guarantees on social welfare {\em provided}
the price increment and stopping rule are well-chosen.
This is very surprising in that the choice of price increment has been used primarily
to adjust auction duration and the stopping rule has attracted little attention.
The main result is a polylogarithmic approximation guarantee for
social welfare when the maximum number of items demanded $\C$ by a bidder is fixed.
Specifically, we show that either the revenue of the CCA is at least 
an $\Omega\Big(\frac{1}{\C^{2}\log n\log^2m}\Big)$-fraction of the optimal welfare or 
the welfare of the CCA is at least an $\Omega\Big(\frac{1}{\log n}\Big)$-fraction of the optimal 
welfare, where $n$ is the number of bidders and $m$ is the number of items. As a corollary,  
the welfare ratio -- the worst case ratio between the social welfare of the 
optimum allocation and the social welfare of the CCA allocation -- is at most $O(\C^2 \cdot \log n \cdot \log^2 m)$.
We emphasize that this latter result requires no assumption on bidders valuation functions.
Finally, we prove that such a dependence on $\C$ is necessary. In particular, we show that the 
welfare ratio of the CCA is at least $\Omega \Big(\C \cdot \frac{\log m}{\log \log m}\Big)$.
\end{abstract}
\notshow{\begin{center}

{\sc N. Bousquet}\footnote{Department of Mathematics and Statistics, McGill University. 
Email: {\tt nicolas.bousquet2@mail.mcgill.ca} },
{\sc Y. Cai}\footnote{School of Computer Science, McGill University
Email: {\tt cai@cs.mcgill.ca} },
{\sc C. Hunkenschr\"oder}\footnote{School of Computer Science, McGill University. 
Email: {\tt chr.hunkenschroeder@gmail.com} },
and
{\sc A. Vetta}\footnote{Department of Mathematics and Statistics,
and School of Computer Science, McGill University. Email: {\tt vetta@math.mcgill.ca} }
\\[1cm]
\end{center}
}
\newpage
\section{Introduction}

The question of how best to allocate spectrum dates back over a century.
In the academic literature, the case in favor of the 
sale of bandwidth was first formalized by Coase \cite{Coa59} in 1959. Since the 1990s spectrum auctions 
have been implemented world-wide.
Moreover, for these bandwidth auctions, it has become apparent that ``not all markets are alike".
Outcomes, in terms of economic efficiency, revenue and the resultant level of competition
in the telecommunications industry, are heavily dependent upon the choice of auction mechanism -- see
\cite{McM94a}, \cite{Kle04} and \cite{Cra13} for detailed discussions.

In practice, two simultaneous ascending price type auctions have proved extremely successful: the
simultaneous multiround ascending auction (SMRA) and the combinatorial
clock auction (CCA). The SMRA was designed by Milgrom, Wilson and McAfee for the 1994 FCC spectrum 
auction; see \cite{Mil04}. Compared with sealed-bid auctions, 
such as the VCG, ascending auctions are widely believed to be more suitable for this scenario.
For example, ascending auctions induce information transfers that allow prices to more accurately reflect 
valuations.\footnote{Indeed, the price discovery process allows bidders to learn valuations (including their own valuations!). 
This is particularly important in bandwidth auctions \cite{Cra13}.} 
A consequence is more economically efficient allocations and, potentially, higher revenues. 
The SMRA has been very successful in practice
but it also has drawbacks \cite{Cra13}. Amongst them is the exposure problem:
a large set may be desired but such a bid may result in being allocated only a smaller undesirable 
subset. In auctions with complementarities, such as the spectrum auctions, this can become a serious issue. 
The CCA, due to Porter et al.~\cite{PRR03}, was designed to overcome this problem, and its usage has gained 
substantial momentum recently. Within the last two years, over ten major spectrum auctions have used 
extensions of the CCA~\cite{ACM06, Cra13, AB14b} and generated approximately 20 billion dollars 
in revenue~\cite{AB14b}.\footnote{These auctions were held worldwide, for 
example, in Austria, Australia, Canada, Denmark, Netherlands, UK, Switzerland, etc.}

Whilst, under certain conditions, there are theoretical explanations for the high social welfare (economic efficiency) produced by
the SMRA~\cite{Mil00,KC82, GS99, FKL11}, 
the performance guarantee of the CCA remains 
elusive. One possible reason for this lack of success is that, upon first examination, no
good welfare guarantees seem achievable for the CCA, even for very simple 
valuation functions. In Section~\ref{sec:example 1 increment}, we show that
the {\em welfare ratio} -- the worst case ratio between the social welfare of the 
optimum allocation and the social welfare of the CCA allocation -- can be as high
as $O(\sqrt{n})$ even for unit-demand bidders. Here $n$ denotes the total number of 
bidders (we will denote by $m$ the number of items).
Moreover, we also show the welfare ratio can be has high as $O(n)$-ratio, even when demand bundles have cardinality at most two.
However, a better selection of price increments, allows us to obtain polylogarithmic upper bounds on the welfare 
ratio for the CCA if the bidders demand is small (e.g. she is only interested in bundles of cardinality 
at most polylogarithmic in the number of items and bidders).
In the Porter et al. CCA mechanism, price increments were fixed to be $1$ whenever there is excess demand.
We obtain our strong welfare guarantees simply by requiring the price increments to be a function of excess demand. 
Specifically, if there are $k$ bids containing item $j$ then
we increase $p_{j}$ by $\epsilon\cdot k$, where $\epsilon$ is a small constant chosen by the auctioneer. 
With this modification, we obtain:
\begin{informal}
If all bids have cardinality at most $\C$, then either the revenue of the CCA is at least an {$\Omega(\frac{1}{\C^{2}\log n\log^2{m}})$}-fraction 
of the optimal social welfare, or the welfare of the CCA is at least an $\Omega(\frac{1}{\log n})$-fraction of the optimal 
welfare.
\end{informal}
This result has two appealing properties. First, it does not make any assumption on the valuations. Hence, it accommodates 
complementarities, which are common in combinatorial auctions 
but are typically hard to deal with. Second, it guarantees that either the 
revenue of the CCA is high or the welfare is close to the optimal. 
Since the social welfare is always not smaller than the revenue, our 
theorem directly implies the following upper bound on the welfare ratio. 
\begin{informal}
If all bids have cardinality at most $\C$ then the welfare ratio for the CCA is at most $O(\C^{2}\log n\log^2{m})$.
\end{informal}
Therefore when bidders only demand sets of cardinality at most $k$, where $k=O(\poly(\log n, \log m))$, the CCA 
has polylogarithmic welfare ratio. 

So the choice of price-increments is fundamental in guaranteeing
good performance. Our results show that the choice of stopping rule 
is also critical. Indeed, they rely upon usage of the original stopping rule
of Porter et al.~\cite{PRR03}. 
However, in recent versions of the CCA
this stopping rule has been replaced by the simpler stopping rule used by the SMRA -- unfortunately,
we show in Section~\ref{sec:stopping} that the SMRA stopping rule is insufficient to guarantee 
good welfare.
Finally, one might wonder if the dependence on $\C$ is necessary. 
We show in Section~\ref{sec:dependence on C} that this 
dependence is unavoidable for general valuations. 

\begin{informal}
For any integer $\C$, there exist arbitrarily large integers $n$ and $m$ and an auction 
with $n$ bidders and $m$ items such that each 
bidder only bid on sets with cardinality at most $\C$ where the welfare ratio of the CCA
 is at least $\Omega(\C\cdot \frac{\log m}{\log \log m})$.
\end{informal}

We now provide a brief road-map of our paper. In Section~\ref{sec:overview}, we
give an overview of the CCA and SMRA auctions and discuss
related work. We give a formal description of the CCA in Section~\ref{sec:CCA}. 
In Section~\ref{app:examples}, we provide examples showing that wrong choices of 
price increment and stopping condition can be detrimental for the performance of the CCA.
In Section~\ref{sec:unit-demand}, we introduce some of 
the key ideas and techniques required in a simple setting. Namely, we
use them to prove a polylogarithmic welfare guarantee
for the special case of unit-demand bidders. 
We generalize this approach to obtain our main results in Section~\ref{sec:general}. 
Finally, we prove a lower bound of welfare approximation for the CCA in Section~\ref{sec:dependence on C}.

\section{An Overview of the SMRA and CCA}\label{sec:overview}
Both the SMRA and the CCA are based upon the
same underlying ascending price mechanism: at time $t$, each item $j$ has a price $p^t_j$. At these prices, each bidder $i$
selects her preferred set $S^t_i$ of items. The price of any item that has excess demand then
rises in the next time period and the process then is repeated. 
The first major difference between these auction mechanisms lies in the bidding language. 
The SMRA uses {\em item bidding}. The auctioneer views the selection 
$S_i^t$ as a collection of bids, one bid for every subset of $S_i^t$. 
However this leads to the exposure problem. To overcome this problem, the CCA uses {\em package (combinatorial) bidding}.
A package bid is an all-or-nothing proposition: the selection of $S_i^t$ is a bid for exactly
$S_i^t$ and nothing else. 

The original CCA, due to Porter et al.~\cite{PRR03}, terminates whenever the last round bids are disjoint 
and not in conflict with the revenue-optimal allocation; see Section~\ref{sec:CCA} for details.
The SMRA (and later versions of the CCA) terminates when there is no excess demand for any item. 
The auctions differ significantly in how items are allocated. 
The SMRA utilizes the concept  of {\em standing high bid}.
Any item (with a positive price) has a {\em provisional winner}. That bidder will win the item unless a higher
bid is received in a later round. If such a bid is received then the standing high bid is increased
and a new provisional winner assigned (chosen at random in the case of a tie). It is not difficult to see that 
this allocation rule increases the risk of exposure for the bidders.
In the CCA, the maximum revenue allocation is output. All bids, regardless of the time they were made, 
are eligible for this allocation, with the constraint that each bidder has at most one accepted bid.
 
 \subsection{Social Welfare}
In practice these ascending auctions have performed extremely well. A major reason for
this is that the associated dynamic processes encourage accurate price discovery \cite{ACM06, Cra13}. 
For the SMRA, there are theoretical results that explain this practical performance.
These results are driven by the use of (aesthetically unappealing but important) standing high bids
which ensure that every item with a positive price is sold, since each such item has a provisional winner. 
This concept of standing high bids was
introduced by Crawford and Knoer \cite{CK81} to study a simple market matching workers to firms.
Their model encompasses an ascending auction with unit-demand bidders which 
converges, under truthful bidding, to a Walrasian equilibrium that maximizes social welfare.
Moreover, Kelso and Crawford \cite{KC82} showed that welfare-maximizing Walrasian equilibrium
are also obtained, under truthful bidding, in auctions where the bidder valuations satisfy the gross substitutes
property; see also \cite{GS99}. 
The method to select items for price rises in the Kelso-Crawford mechanism differs from the 
more natural choice made by the SMRA, which increments the prices of all items under excess demand. 
Milgrom \cite{Mil00}, however, showed these results continue to hold for the SMRA.
Walrasian equilibrium need not exist for more general valuations, but
approximate welfare guarantees can still be obtained for submodular valuation functions~\cite{FKL11}.
No such theoretical results are know for the CCA and that is the motivation behind this work. 
 
\subsection{Bidding Activity Rules}
It is important to note that, in practice, accompanying the 
ascending price mechanisms are a set of {\em bidding activity rules}. 
The activity rules are designed to induce truthful bidding in each round. This is extremely important -- Cramton \cite{Cra13} 
states that the ``{\em truthful expression of preferences is what leads to excellent price discovery and ultimately
an efficient auction outcome}''. For the SMRA, though, the activity rules are quite weak
and strategic bidding is common and, often, profitable \cite{Cra13}. In contrast, the 
CCA incorporates a much stronger set of bidding activity rules for each round than 
the SMRA. In particular, the CCA applies a set of {\em revealed preference} (RP)
constraints on feasible bids. Suppose at time $s$ we bid for package $S$ and at time $t > s$ we bid for package $T$. 
In its weak form, see Ausubel et al. \cite{ACM06}, revealed preference
produces a constraint $p^t(S)-p^s(S) \ge p^t(T)-p^s(T)$. That is, such bidding behavior can 
only be rational if
the price of package $S$ has risen by at least the rise in the price of package $T$. If not the 
bid for package $T$ is forbidden. Moreover, in its general form \cite{HBP10, AB14}, the revealed preference 
rules ensure that the only bidding strategies that are admissible correspond to virtual valuation 
functions. Indeed, even relaxed implementations of the constraints 
allow only for approximate virtual valuation functions; see Boodaghians and Vetta \cite{BV15}.
If follows that strategic behavior must take a very restricted form in the CCA, essentially
consisting of pre-commiting to a virtual valuation function. 
Given the lack of information and the dynamic nature of the CCA, such a pre-commitment is 
hard to compute and extremely risky. Consequently, the working assumption in this paper that
bidders behave truthfully in a competitive CCA auction is reasonable.


\subsection{Practical Implications}
We show that modifying the price-increments can have fundamental impact on social welfare.
This is quite remarkable because, in practice, the choice of price increments has been primarily considered
 as a matter of fine-tuning. Price increments are seen as a way to affect the 
length of the auction whilst having minimal effect on the final outcome.
Indeed, Ausubel and Baranov \cite{AB14b} state that
``{\em Among all design decisions that need to be made prior to the auction, [the choice of price increments] is
 considered relatively unimportant and is often overlooked by the design team}''.
 Our results show that the choice of price increments is actually extremely important. 

Our results also show that another apparently innocuous aspect of the CCA mechanism
is vital in generating high welfare: the choice of {\em stopping rule}.
Recall, the original CCA only terminates when there is no excess demand induced by the bids in the current round
{\em and} the maximum revenue allocation over all rounds is not in conflict with the current round bids. Current implementations
of the CCA are based upon the two-stage mechanism of Ausubel, Cramton and Milgrom \cite{ACM06}.
There, the ascending price mechanism, as described, is used in the first-phase except
that the stopping rule reverts to the simple rule used in the SMRA: the mechanism
terminates if bids in the current round are disjoint. The use of this simple stopping rule is
unfortunate: its use cannot guarantee high welfare. We present,
in Section \ref{sec:stopping}, a simple example with demand sets of cardinality two where, under this stopping rule,
the CCA produces arbitrarily poor welfare, even when price-increments are a function of excess demand.

Finally, it remains to discuss computational aspects. The combinatorial 
allocation problem is notoriously hard to approximate (it contains maximum independent 
set as a special case). This appears to suggest that the CCA is not implementable in polynomial time.
In spectrum auctions, however, bid patterns are highly structured.
This has two major effects. First, given a set of prices, bidders do seem able to make bids in a timely manner.
Second, the resultant combinatorial allocation problems can be solved almost instantaneously 
using standard branch and bound optimization techniques. As a consequence, it does not appear
that computational constraints are currently a major concern in implementing ascending auctions
in practice.

\subsection{Related Work}
An alternative approach for unit-demand auctions 
was examined by Demange, Gale and Sotomayor \cite{DGS86}.
Their ascending auction also outputs a Walrasian equilibrium that maximizes social 
welfare but without the need for standing high bidders.
To achieve this, however, each bidder 
is now required to submit their entire demand set\footnote{The demand set consists of every bundle that
maximizes profits given the prices in that round.} 
 in each round, rather than
just a single bid as in the SMRA and CCA. Given prices, the mechanism tests whether a Walrasian
equilibrium can be produced from  the demand sets; if not, a set of items under excess demand is obtained
based upon Hall's theorem.
Gul and Stacchetti \cite{GS00} showed this approach also generalizes to auctions where
bidder valuations satisfy the gross substitutes property.
Interestingly, these ascending auctions produce the minimum set of Walrasian prices, which in the case of
unit-demands also correspond to VCG payments. However, Gul and Stacchetti \cite{GS00} 
proved that it is not possible to implement the VCG mechanism 
via an ascending auction for general valuations, even those with the gross 
substitutes property. Thus, truthful bidding does not form an equilibrium in the corresponding direct mechanism.
To go beyond the gross substitutes property, De Vries, Schummer and Vohra \cite{VSV05} 
dropped the requirement of anonymous prices. Instead, each bundle requires a separate price for
each bidder.

From the theory of computation side, a relevant and fruitful direction is to approximate the welfare of a 
combinatorial auction with simple but not necessarily truthful 
auctions~\cite{ChristodoulouKS08,BhawalkarR11,HassidimKMN11,FeldmanFGL13, ST13, BLN14}, for example, simultaneous 
single-item auctions, the ascending auction of Gul and Stacchetti, etc.
The focus of this sequence of papers is to show the price of 
anarchy in these games is some small constant for a certain set of 
valuations. In particular, Babaioff et al.~\cite{BLN14} showed that as long as the simple auction maximizes the welfare over 
the bidders' declared valuations, the price of anarchy is a small constant. Unfortunately, this result 
cannot be applied to the CCA, which is not a declared welfare maximizer. Indeed, to the best of our knowledge, 
none the known results directly applies to the CCA. We suspect that the slow progress is largely due to the ignorance 
of the complicated dynamic behavior of the mechanism. 
We hope the results and techniques developed may help tackle other research questions related to the CCA, such as its
price of anarchy.

\section{The Combinatorial Clock Auction}\label{sec:CCA}
We now give a detailed description of the CCA. To begin with, we present some notations and definitions.
Let $N$ be the set of the bidders and $M$ be the set of items. Each bidder $i$ has value $v_{i}(S)$ 
for any set of items $S\subseteq M$. The {price $p(S)$} for a set of item $S$ is simply the sum of 
prices for each item in~$S$. The utility of bidder $i$ for set $S$ is $v_{i}(S)-p(S)$, where $p(S)$ is 
the price of $S$. The CCA outputs an allocation $\vec{S}=\{S_{1},\ldots, S_{n}\}$, where the $S_{i}$ are pairwise-disjoint subsets of the items, 
and a collection of prices $\{p_{j}\}_{j\in[m]}$. The \emph{social welfare} of an 
allocation $\vec{S}$ is $\sum_{i}v_{i}(S_{i})$ and the \emph{revenue} 
is $\sum_{i} p(S_{i})=\sum_{i}\sum_{j\in S_{i}}p_{j}$. Let $\Opt=\max_{\vec{S}}\sum_{i}v_{i}(S_{i})$ be the 
optimal social welfare and $\vec{R}^{*}$ be the corresponding allocation. The CCA is then formalized in Procedure \ref{alg1}. 

\begin{algorithm}
\caption{The Combinatorial Clock Auction}
\label{alg1}
\begin{algorithmic}[PERF]
\STATE Let $t=0$ and the initial price $p_{j}^{0}$ be $0$ for every item $j$. 
\LOOP
      \STATE Each bidder $i$ bids for the set of items $S^{t}_{i}$ of largest positive utility (breaking ties arbitrarily).\\
      A bidder drops out if she has non-positive utility for every subset of $M$. \\
      Let $P_{i}^{t}=\sum_{j \in S_{i}^{t}} p_j^t$ be the price for set $S_{i}^{t}$ in round $t$.
      \FOR{$j=1$ to $n$}
	\STATE $p_j^{t+1} \leftarrow  p_j^t\ + {\epsilon \cdot \sum_{S_{i}^{t}: j\in S^{t}_{i}} 1}$.
      \ENDFOR
  \IF{{the sets $S_i^{t}$ are pairwise disjoint}}
  \STATE Amongst all the bids $\{ (S_{i}^{t'}, P_{i}^{t'}) \}_{i\in [n],t'\leq t}$ ever made, find the revenue 
  maximizing allocation $\vec{S}^{*}=(S_{1}^{t_{1}^{*}},\ldots,S_{n}^{t_{n}^{*}})$, that 
  is, $\vec{S}^{*}\in \argmax_{\text{disjoint sets: }S_{1}^{t_{1}},...,S_{n}^{t_{n}}} \sum_{i} \sum_{j\in S_{i}^{t_{i}}} p_{j}^{t_{i}}.$
\IF {{$S_i^{t} \cap S_j^{t_j^{*}} = \varnothing$ for every $i \neq j$} }
\STATE {\bf Output} Allocate the set of items $S_{i}^{t_{i}^{*}}$ to bidder $i$, and charge 
her $P_{i}^{t_{i}^{*}}$ as the payment.
\ENDIF
\ENDIF
      \STATE $t \leftarrow t+1$.

\ENDLOOP
\end{algorithmic}
\end{algorithm}
Note that the CCA does terminate. If not, prices will monotonically increase
and eventually force every bidder to drop out, then the stopping condition is trivially satisfied. The length of the auction  depends on $\epsilon$. In most ascending auctions, the price increment $\epsilon$ is a small constant chosen by the 
auctioneer and the values of bidders are assumed to be 
integers. Since the CCA is scale invariant\footnote{If we multiply all the values, bids and prices by a common factor, 
the execution of the mechanism is the same.}, for notational simplicity, we set $\epsilon=1$ and 
assume every bidder $i$'s value for any set of items $S$ is a multiple of some integer $W\geq n^{3}m^{2}$. 



Before proceeding further, we present a few simple but useful facts about the CCA. 
The {\em utility} of bidder $i$ in round $t$, denoted by $u_i^t$, is her utility for her favorite set of items, given
the prices in round $t$. Formally, $u_{i}^{t} = \max _{S {\subseteq} M} \left(v_{i}(S)-\sum_{j\in S} p_{j}^{t}\right)$. Since $S$ 
can be $\emptyset$, $u_{i}^{t}$ is always non-negative. 
With this definition, we are ready to show some  properties of the CCA. As Fact~\ref{fact:v geq u} is self-evident, 
we do not provide the proof here.


\begin{fact}\label{fact:monotone utility}
For any bidder, $u_{i}^{t}$ is monotonically {non-increasing}.
\end{fact}
 \begin{proof}
 Suppose $u_i^t <u_i^{t'}$ for some $t'>t$. Let $S$ be a set of items satisfying $u_i^{t'}=v_i(S)-\sum_{j\in S} p_{j}^{t'}$. 
 Since prices are non-decreasing, we have $v_i(S)-\sum_{j\in S} p_{j}^{t} \geq v_i(S)-\sum_{j\in S} p_{j}^{t'} > u_{i}^{t}$, contradicting the maximality of $u_i^t$.
 \end{proof}

\begin{fact}\label{fact:v geq u}
If bidder $i$ bids on $S$ in round $t>0$, then $v_{i}(S) {>} u_{i}^{t}$.
\end{fact}


\begin{fact}\label{fact:disjoint bids}
{If bidder $i$ is still active when the stopping condition is met, then $i$ is allocated a subset of items whose value 
is at least her utility in the final round.}
\end{fact}
 \begin{proof}
 Let $\hat{t}$ be the final round. 
  Since no item in $S_i^{\hat t}$ is allocated to any bidder $j \neq i$ in the CCA (by definition of 
 the stopping condition), the entire set $S_i^{\hat t}$ may still be allocated by the mechanism to bidder $i$. 
 Thus, $i$ must win some items in the revenue-optimal allocation. Let us assume she wins $S_{i}^{t_{i}^{*}}\neq \emptyset$ for some $t_{i}^{*}\leq \hat t$.
 By Fact~\ref{fact:monotone utility}, $u_{i}^{t_{i}^{*}}\geq u_i^{\hat t}$.
Therefore, her value for  $S_{i}^{t_{i}^{*}}$ is clearly at least $u_i^{\hat t}$.
 \end{proof}

\section{Social Welfare Guarantees for the CCA}\label{sec:CCA analysis}
We are now ready to quantitatively analyse the CCA. We begin with the case of unit-demand bidders, 
and prove that the CCA achieves a polylogarithmic fraction of the optimal social welfare. 
Whilst the case of unit-demand bidders might seem limited, the 
techniques developed for this basic case will be important as we then use them to 
extend our results to general valuation functions.

\subsection{The Welfare Ratio for Unit-demand Bidders}\label{sec:unit-demand}
We say a bidder is \emph{unit-demand}, if she demands one item at most. Formally, we define bidder $i$'s valuation as 
$v_{i}(S)=\max_{j\in S} v_{ij}$ where $v_{ij}$ is $i$'s value for item $j$. 
For unit-demand bidders, a feasible allocation is simply a matching between the bidders 
and the items. In this section, we show the matching selected by the CCA achieves an 
$\Omega\Big(\frac{1}{\log n \log^{2} m}\Big)$-fraction of the optimal social welfare. 
\textcolor{purple}{\notshow{In fact, our result is even stronger. Namely, we show that the 
revenue of the CCA is at least an $\Omega\Big(\frac{1}{\log n \log^{2} m}\Big)$-fraction 
of the optimal social welfare.}} But we have already seen two
mechanism that maximize social welfare in this special case. So before proving the logarithmic welfare ratio
for the CCA, it is informative to understand
why the CCA does not achieve optimality.
First the Crawford-Knoer mechanism \cite{CK81} achieves optimality via the use of standing high bids.
But the motivation behind the CCA was to allow package bidding on multi-item auctions
with complementarities and then standing high bids. 
The Demange et al. mechanism \cite{DGS86} achieves optimality by requiring 
that each bidder submits her entire demand set. 
From a theoretical viewpoint that is exactly the right thing to do, and the CCA losses out by
requiring one bid per round only.  In practice, however, the Gul and Stacchetti mechanism \cite{GS00} 
(which generalizes the Demange et al. mechanism for general demand) would be extremely complicated 
for bidders to use.
Moreover, it is not clear how one could use simple bidding activity rules to incentivize truthful bidding
in such a complex auction. In contrast, the CCA is a very simple mechanism
that is incentivizable using bidding activity rules.
This is important as experiments~\cite{BGM14} suggest that simplicity is key if we want to
generate welfare and revenue.

Now let's return to analysing the CCA. Whilst it is difficult to directly relate the welfare of the CCA with the optimal social welfare, 
we establish our result using a greedy allocation as a proxy of the CCA's outcome. We show that 
there are only two possibilities for the greedy allocation: (i) the revenue of this allocation is  
at least $\Omega\Big(\frac{\Opt}{\log n \log^{2} m}\Big)$, and since the CCA selects the 
revenue optimal allocation, its revenue can only be higher; or (ii) {the greedy allocation has 
small revenue, but many bidders still have high utility when the ascending-price phase ends. 
Combining this property with Fact~\ref{fact:disjoint bids}, we can immediately show that the 
welfare of the allocation selected by the CCA is at least $\Omega\Big(\frac{\Opt}{\log n}\Big)$.

The greedy allocation method is shown in Procedure~\ref{alg2}. Since bidders are unit-demand, we will
 use $v_{ij}$ to denote bidder $i$'s value for item $j$, and $s_{i}^{t}$ to denote the item 
 bidder $i$ bids on at round $t$.  Our greedy algorithm simply allocates the item to the highest available bid, 
 then removes all bids that conflict with it and repeats. We terminate this procedure when the highest bid is smaller 
 than some predetermined threshold $b$ only to make the analysis cleaner.   
 \textcolor{purple}{\sout{We show that the revenue of the greedy algorithm is at 
 least $\Omega(\frac{1}{\log n\log^{2}m})$-fraction of the optimal social welfare.}}

\begin{algorithm}
\caption{Greedy Allocation Procedure for Unit-demand Bidders}
\label{alg2}
\begin{algorithmic}[PERF]
\STATE {\bf Input:}
$\mathcal{S} = \{(s_{i}^{t},p_{s_{i}^{t}}^{t})\}_{i,t}$ the collection of bids made in the CCA 
and $b\geq n^{2}$ the threshold.
\WHILE{$\mathcal{S} \neq \emptyset$}
   \STATE Let $(s_{i}^{t},p_{s_{i}^t}^{t})$ be the bid with maximum price (break ties arbitrarily).
   \IF {$p_{s_{i}^t}^{t} \geq b$}
   \STATE Allocate item $s_{i}^{t}$ to bidder $i$ with price $p_{s_{i}^t}^{t}$. \\
   Remove every bid of bidder $i$ and remove every bid (or any bidder) for item $s_{i}^{t}$ in $\mathcal{S}$.
       \ELSE \RETURN
   \ENDIF
\ENDWHILE
\end{algorithmic}
\end{algorithm}

\notshow{\begin{theorem}\label{thm:matching}
\textcolor{purple}{\sout{
There exists a threshold $b$, such that the allocation found by Algorithm~\ref{alg2} 
has at least revenue {$\frac{\Opt}{{120}\log n\log^2{m}}$}.}}
\end{theorem}}

\textcolor{purple}{\sout{With Theorem~\ref{thm:matching}, it is obvious to argue the CCA 
has high revenue and social welfare.}}

\begin{theorem}\label{thm:matching CCA}
Either the revenue of the CCA for unit-demand bidders is at least {$\frac{\Opt}{{480}\log n\log^2{m}}$} 
{or the social welfare is at least $\frac{OPT}{24 \log n}$.}~\footnote{We remark that we have not attempted to optimize 
the constants in this theorem.} Thus, the welfare ratio is at most $O(\log n \log^2 m)$. 
\end{theorem}

\notshow{\begin{proof}
Since the CCA selects the allocation that maximizes the revenue, its revenue is no less than the revenue of 
the greedy algorithm. Because the social welfare of the CCA must be at least as much as its revenue, we 
prove our claims.
\end{proof}}

We will use the following notation. Let $X$ be the set of bidders that have been assigned items in the 
greedy algorithm and $\tilde{X}$ be the set of items that are allocated to them. Let $k\geq |X|$ be some integer 
and $c\in [3,{\frac{n}{2}-1}]$ be an integer that we will specify later. A key lemma is that if $k$ is small then the utility of 
bidders that are not in $X$ decreases in a slow rate. Specifically,

\begin{lemma}[Time Amplifying for Unit-demand Bidders]\label{lem:time amplifying}
Let $S$ be a set of at least $c\cdot k$ bidders disjoint from $X$ ($|X|\leq k$), such that every bidder $i\in S$ 
has utility at least $u\geq 2b$ in round $t\geq b+c-1$. If the greedy algorithm has revenue less than 
$k\cdot b-mn$, then in {any round up to} $(c-2)(t+1)-1$ (the mechanism can terminate before that), 
there is a subset of at least $|S|-c\cdot k$ bidders of $S$ such that each of them has utility at least $u-2b$.
\end{lemma}

As Lemma~\ref{lem:time amplifying} is mainly used when $|S|>>c\cdot k$ and $u>> b$, let us 
consider $|S|$ and $u$ being much larger than $k$ and $b$. Intuitively, it states that if at round $t$ 
there exists a large set $S$, in which all bidders have high utility, then throughout round $c\cdot t$, 
most bidders in $S$ (at least $|S|-c\cdot k$ bidders) still have high utility ($u-2b$). 

Before providing the formal proof of Lemma~\ref{lem:time amplifying}, let us first sketch the underlying idea. Notice that every 
bidder in $S$ has utility at least $u$ till round $t$. Thus, by Facts~\ref{fact:monotone utility} and~\ref{fact:v geq u}, 
every bidder $i$ in $S$ has value at least $u$ for each item in $Q_{i}$, where $Q_{i}$ is the collection of every item that 
$i$ ever bid on in the first $t$ rounds. For any set $S'\subseteq S$, if all bidders in $S'$ have utility no more 
than $u-2b$ at a certain round, then every item in $\cup_{i\in S'}Q_{i}$ must have price at least $2b$. However, 
for any item in $\cup_{i\in S'} Q_{i}-\tilde{X}$, only bidders selected by the greedy algorithm (the ones in $X$) 
can bid on it after the price has reached $b$. If not, then there exists a bidder outside $X$ that should have been 
chosen by the greedy algorithm because she made a bid on some item not in $\tilde{X}$ at price higher than $b$. 
If $|S'|\geq c\cdot k$, then it is easy to argue that $|\cup_{i\in S'} Q_{i}-\tilde{X}|\approx \frac{c\cdot t\cdot k}{b}$. 
As each item in $\cup_{i\in S'} Q_{i}-\tilde{X}$ requires about $b$ bids from bidders in $X$, these 
bidders need to make roughly $c\cdot t\cdot k$ many bids. There are at most $k$ bidders in $X$ so, in total, this will 
take at least $c\cdot t$ rounds. We are now ready to present the formal argument.
 
\vspace{.1in}
\begin{prevproof}{Lemma}{lem:time amplifying}
Let $t'$ be the first round where at least $c\cdot k$ bidders from $S$ have utility at most $u-2b$, 
and let $S'$ be the set of these bidders. If $t'$ does not exist, then at 
least $|S|-c\cdot k$ bidders from $S$ have utility at least $u-2b$ when the mechanism 
terminates. Lemma~\ref{lem:time amplifying} holds.

Because every bidder $i\in S$ has utility $u$ at round $t$, Facts~\ref{fact:monotone utility} 
and~\ref{fact:v geq u} ensure that $v_{ij}\geq u$ for any item $j$ that $i$ bids on in the 
first $t$ rounds. Let $M'$ be the items in $M-\tilde{X}$ that are bid on by some bidder 
in $S'$ during the first $t$ rounds. Bidders in $S'$ make $|S'|\cdot t$ bids in total. How 
many bids can they make on $\tilde{X}$? No bidder from $S'$ is allocated an 
item by the greedy algorithm. Therefore, for any item $j$ in $\tilde{X}$, none of the bids made by 
bidders in $S'$ can exceed the price $p_{j}$ that the greedy algorithm sells the item for. 
So, before the final round, the total number of bids from $S'$ on item $j$ is less than $p_{j}$. In the last round,
they make at most $|S'|\leq n$ bids. Therefore, the 
total number of bids made on $\tilde{X}$ is at most the revenue of the greedy allocation 
plus $n\cdot m$, which is at most $k\cdot b$ by assumption.

 As none of the bidders from $S'$ is selected by the greedy algorithm, bidders from $S'$ 
 must stop bidding on any item in $M-\tilde{X}$ after its price reach $b$.
Thus, the total number of bids made by bidders from $S'$  on any item in $M-\tilde{X}$ is at 
most $b+n$, implying $|M'|\geq \frac{|S'|\cdot t-k\cdot b}{b+n}$. Note that for each 
item in $M'$, at least one bidder in $S'$ has value at least $u$ for it. 
Thus, at round $t'$, the price for each of these items must be at least $2b$, otherwise that bidder 
will have utility greater than $u-2b$. 

In the round when the price of $j$ passes $b$, its price is at most $b+n$. Bidders from $X$ need to 
make at least another $b-n$ bids to drive the price up to $2b$. Because
$|M'|\geq \frac{|S'|\cdot t-k\cdot b}{b+n}$, they must make at least 
${(b-n)}\cdot \frac{\left(|S'|\cdot t-k\cdot b\right)}{b+n}$ bids. As there are at most $k$ bidders in $X$, we have


  \begin{align*}
  t'&\ \geq\  \frac{b-n}{b+n}\cdot \frac{|S'|\cdot t-k\cdot b}{k}\\
  &\ \geq\ (1-\frac 2n)\cdot \frac{|S'|\cdot t-k\cdot b}{k} &&\hspace{-70pt}\quad (b\geq n^{2})\\
  &\ \geq\  (1-\frac 2n)\cdot (c\cdot t-b)			 &&\hspace{-70pt} \quad(|S'|\geq c\cdot k)\\
  &\ \geq\  (1-\frac 2n)\cdot (c-1)(t+1)			 &&\hspace{-70pt} \quad(t\geq b+c-1)\\
  &\ \geq\  (c-2)(t+1) 				 &&\hspace{-70pt} \quad {(\frac{n}{2}-1 \geq c)}
  \end{align*}

By definition of $t'$, it is then straightforward to see that at round $t'-1\geq (c-2)(t+1)-1$ 
there are at least $|S|-c\cdot k$ bidders from $S$ such that each of them has utility at least $u-2b$.
 \end{prevproof}
 
In order to prove Theorem~\ref{thm:matching CCA}, we need one more Lemma. The proof of this Lemma is fairly standard. 
However we prove it for the sake of completeness.
 \begin{lemma}\label{lem:equal value group}
Let $\vec{R}^{*}=(R_{1},\ldots,R_{n})$ be the allocation that maximizes the social welfare.
Then there exists a set $B$ and a real number $v^{*}$ such that every bidder $i\in B$ has 
value between $[v^{*}, 2v^{*}]$ for the set $R_{i}$ and the total value for bidders in $B$ is 
at least $\frac{\Opt}{3\log n}$. 
\end{lemma}
\begin{proof}
Let there be $\ell$ non-empty sets in $\vec{R}^{*}$, and we define the value for a set $R_{i}$ to be $v_{i}(R_{i})$. By definition, $\sum_{i=1}^{\ell} v_{i}(R_{i})=\Opt$. Now we construct $2\lceil \log n\rceil+1$ bins, for any $1\leq i\leq 2\lceil \log n\rceil$ the $i$-th bin $B_{i}$ contains all sets with values in $(\frac{\Opt}{2^{i}},\frac{\Opt}{2^{i-1}}]$ . The last bin contains all the other sets. Since every set in the last bin has value at most $\frac{\Opt}{n^{2}}$ and contains at most $n$ sets, the total value of the last bin is at most $\frac{\Opt}{n}$. Therefore, $\sum_{i=1}^{2\lceil \log n\rceil} v(B_{i}) \geq (1-1/n)\Opt$, where $v(B_{i})$ is the sum of the values from sets in $B_{i}$. That means there exists a bin $B_{i^{*}}$ such that $$v(B_{i^{*}})\geq \frac{(1-1/n)\Opt}{2\lceil \log n\rceil}\geq \frac{\Opt}{3\log n}.$$ We simply let $B=B_{i^{*}}$.
\end{proof}

 With Lemma~\ref{lem:time amplifying} 
and~\ref{lem:equal value group}, we are ready to prove Theorem~}{\ref{thm:matching CCA}}. 
Again, we first sketch the proof idea. Let $B$ and $v^{*}$ be the same as in Lemma~\ref{lem:equal value group}. 
We set the threshold $b$ to be $\Theta(\frac{v^{*}}{\log m})$. The goal is to argue that if the revenue of 
the greedy allocation is less than $\Theta(\frac{\Opt}{\log n\log^{2} m})$, then most bidders from $B$ still 
have utility as high as $\Theta(v^{*})$ when the CCA terminates. This implies the social welfare is at 
least $\Theta(\frac{\Opt}{\log n})$ by Fact~\ref{fact:disjoint bids}. First note that if more than $\frac{B}{\log m}$ 
bidders are selected by the greedy algorithm (\emph{i.e.} bidders in $X$) then the revenue already meets our target. 
So we assume this is not the case. This means the number of bidders in $B-X$ is at least $\log m$ times 
more than the number of bidders in $X$. Since all these bidders have utility at least $v^{*}\approx \log m\cdot b$ in the 
beginning, we can apply Lemma~\ref{lem:time amplifying} to this set of bidders. The key insight in this proof is 
that we can repetitively apply Lemma~\ref{lem:time amplifying} on the set of bidders that still have high utility. 
Because each application of Lemma~\ref{lem:time amplifying} decreases the size of the set of bidders and 
their utilities linearly in $|X|$ and $b$, we can apply it $\Theta(\log m)$ times provided the CCA has not yet 
terminated. If this is the case, the CCA runs for at least $\Theta(m\cdot b)$ rounds. Since any active bidder makes 
a bid each round, it is easy to show that throughout $\Theta(m\cdot b)$ many rounds this bidder has to make a 
bid on some item outside $\tilde{X}$ (the set of items allocated in the greedy algorithm) at a price larger than $b$. 
This gives a contradiction. Hence, the CCA must terminate before that amount of rounds. In this case, we can argue that at 
least a constant fraction of bidders in $B-X$ still have utility $\Theta(v^{*})$ when the CCA terminates, implying 
high social welfare. We now formalise this argument.
\vspace{5pt}

\begin{prevproof}{Theorem}{thm:matching CCA}
Let $B$ and $v^{*}$ be as in Lemma~\ref{lem:equal value group}.
We use $\mathcal{K}$ to denote the size of $B$. Let the threshold $b=\frac{v^{*}}{c_{1}\log m}$ 
and $k=\frac{\mathcal{K}}{c_{2}\log m}$ for some constants $c_{1}$ and $c_{2}$ that will be 
specified later. We will prove that either the revenue of the greedy algorithm, with our choice 
of $b$, is at least $g=\frac{\Opt}{6c_{1}c_{2}\log n \log^{2}m}$ or the social welfare of the 
CCA is at least $\frac{\Opt}{24\log n}$.

First, by Lemma~\ref{lem:equal value group} and our choice of the parameters, it is easy to 
verify that $k\cdot b\geq g$. 
So we can assume $|X|\leq k$ from now on. Otherwise the revenue is already at 
least  $|X|\cdot b>g$. Let $X'$ be the set of bidders that items in $\tilde{X}$ are 
allocated to under $\vec{R}^{*}$ and let $Y=B-X-X'$. Because $|X'|\leq |\tilde{X}|=|X|\leq k$, 
we have $|Y|\geq (c_{2}\log m-2)k$.
Let $Z$ be the set of items that are allocated to bidders in $Y$ in the allocation $\vec{R}^{*}$, 
and let $t_{1}$ be the first round {(if it exists)} where at least $c\cdot k$ items in $Z$ have prices 
at least $2b$. We proceed by case analysis.

\vspace{0.05in}
\noindent\textbf{Case (1):} $t_1$ does not exist. By the definition of $t_{1}$, this means that when 
the CCA terminates at round $t$,
there exists a set $Y'\subseteq Y$ of bidders whose allocated items in $\vec{R}^{*}$ have 
prices no more than $2b$ at round $t$ and $|Y'| \geq |Y|-ck\geq |Y|/2$.
Therefore, $u_i^{t} \geq v^*-2b\geq v^{*}/2$ for any bidder $i\in Y'$. By Fact~\ref{fact:disjoint bids}, 
the social welfare of the greedy allocation is then at least 
\[|Y'|\cdot u_{i}^{t}\geq\frac{|Y| \cdot v^*}{4} \geq \frac{c_2 \cdot k \cdot \log m \cdot v^*}{8} 
\geq \frac{\mathcal{K} v^*}{8} \geq \frac{OPT}{24 \log n}.\]

\noindent\textbf{Case (2):} $t_1$ exists. Since $Z$ is disjoint from $\tilde{X}$, only bidders 
from $X$ can bid on items in $Z$ after their prices reach $b$. Thus, $t_{1}\geq \frac{(b-n)\cdot c\cdot k}{|X|}$, 
because $k\geq |X|$ and $b\geq n^{2}$ {and $c \geq 3$}, $t_{1}\geq c\cdot (b-n) >b+c$.
By the definition of $t_{1}$, we know at round $t_{1}-1$ there are at least $|Y|-c\cdot k$ bidders whose
 allocated items in $\vec{R}^{*}$ have prices less than $2b$. Let us call this set of bidders $Y_{1}$, and 
 clearly they all have utility at least $u_{1}= v^{*}-2b$ at round $t_{1}-1$. Let $t_{2}$ be the first 
 round (if it exists) that at least $c\cdot k$ bidders from $Y_{1}$ have utility no greater than 
 $u_{2}=u_{1}-2b$, and let $Y_{2}$ be the set of bidders from $Y_{1}$ that still have utility at 
 least $u_{2}$ at round $t_{2}-1$. 
Now we let $Y_{1}$ be $S$, $u_{1}$ be $u$ and $t_{1}-1$ be $t$ and apply Lemma~\ref{lem:time amplifying} on 
them, the lemma gives that $t_{2}-1\geq(c-2)t_{1}-1$. Hence, $t_{2}\geq (c-2)\cdot t_{1}$. 

There is nothing special about $Y_{1}$, $u_{1}$ and $t_{1}-1$. If we recursively define $t_{i}$ as the first round 
that at least $c\cdot k$ bidders from $Y_{i-1}$ have utility no greater than $u_{i}=u_{i-1}-2b$, and $Y_{i}$ as the 
set of bidders from $Y_{i-1}$ that still have utility at least $u_{i}$ at round $t_{i}-1$, we can apply 
Lemma~\ref{lem:time amplifying} on $Y_{i-1}$, $u_{i-1}$ and $t_{i-1}-1$ as long as they satisfy 
the conditions in Lemma~\ref{lem:time amplifying} (and that $t_{i}$ exists). In that case, we 
have $t_{i}\geq (c-2)\cdot t_{i-1}$. How many times can we apply Lemma~\ref{lem:time amplifying} 
before the conditions are violated? 
Since the size of $Y_{i}$ decreases by at most $c\cdot k$ and $u_{i}$ decreases by $2b$ in every 
application of Lemma~\ref{lem:time amplifying}, we have $|Y_i| \geq {|Y|}-i \cdot c \cdot k$ and $u_i={v^*}-2i \cdot b$. 
To violate the conditions of Lemma~\ref{lem:time amplifying}, we need $|Y_i| < c\cdot k$ or $u_i < 2b$, 
hence we can apply Lemma~\ref{lem:time amplifying} for at least 
${\ell'}=\min\{[\frac{{|Y|}}{c\cdot k}], [\frac{{v^*}}{2b}] \}-1\geq \min\{[\frac{c_{2}-{1}}{c}\cdot \log m], [\frac{c_{1}-1}{2}\cdot \log m]\}$ times. 
If we let {$c_{1}=8$, $c_{2}=10$} {and $c=4$}, we have ${\ell' \geq 2 \log m+6}$. Remember that we might not 
be able to run the recursion till the conditions are violated, because the CCA might terminate before that. We 
now use case analysis to show that if we take $\ell=\log m+2$, then our claim holds no matter whether 
$t_{\ell}$ exists or not. 

\noindent\textbf{Case (i):} The CCA terminates between $t_{j}$ and $t_{j+1}$ for some $j<\ell$. In this case, 
there are at least $|Y_{j}|-c\cdot k$ bidders each of whom has utility at least $u_{j}-2b$. 
{Since $\ell \leq \frac{\ell'}{2}+1$, we have $|Y_{j}|-c \cdot k \geq |Y|-c(j+1)k \geq |Y|/2$ 
and $u_j-2b \geq v^* - 2(j+1)b \geq v^*/2$. Thus, as in Case (1), the social welfare of the 
greedy allocation is at least $\frac{OPT}{24 \log n}$.} 

\noindent\textbf{Case (ii):} $t_{\ell}$ exists. In this case, we argue that there is a bidder who has made a bid 
larger than $b$ on $M \setminus X'$ that has not been selected by the greedy algorithm. This results in a 
contradiction. Note that $t_{\ell}\geq 2^{\ell-1}\cdot t_{1}>2m\cdot t_{1}>(k+m)\cdot b$. The last inequality 
holds because every bidder in $B$ receives an item, so $k\leq \mathcal{K}\leq m$. 
Hence, there is a bidder $i$ in $Y_{\ell}$ that has made $(k+m)\cdot b$ bids in the ascending-price phase. 
The revenue of the greedy allocation is less than $k\cdot b$, so $i$ can make at most $k\cdot b$ bids 
on items in $\tilde{X}$. Therefore, $i$ makes at least $m\cdot b$ bids on $M-\tilde{X}$, which means there 
must be one item $j$ in $M-\tilde{X}$ that $i$ has bid on with price larger than $b$. This cannot happen, 
otherwise $i$ would have been selected by the greedy algorithm.
\end{prevproof}

\subsection{An Upper Bound on the Welfare Ratio for General Bidders}\label{sec:general}
In this section, we generalize Theorem~\ref{thm:matching CCA} to accommodate general valuation functions. The idea is similar.
Using a greedy algorithm as a proxy for the CCA, we argue that there are only two possibilities: (i) the 
revenue of the greedy allocation is  at least $\Omega\Big(\frac{\Opt}{\C^2\log n \log^{2} m}\Big)$, and since 
the CCA selects the revenue optimal allocation, it must achieve no less revenue than the greedy allocation; 
or (ii) the greedy allocation has small revenue, but many bidders still have high utility at the final round. 
Using Fact~\ref{fact:disjoint bids}, we can immediately show that the welfare of the allocation selected 
by the CCA is at least $\Omega(\frac{\Opt}{\log n})$.

Let us first specify the greedy algorithm.

\begin{algorithm}
\caption{Greedy Allocation Procedure for General Bidders}
\label{alg3}
\begin{algorithmic}[PERF]
\STATE {\bf Input:} $M$ is the set of items. $N$ is the set of bidders. $\mathcal{S} = \{(S_{i}^{t}, P_{i}^{t})\}_{i,t}$ is the collection of bids made in the CCA. $b\geq n^{2}$ is a threshold on the bid price.
\WHILE{$\mathcal{S} \neq \varnothing$}
   \STATE Let $(S_{i}^{t},P_{i}^t)$ be a bid of maximum price (break ties arbitrarily).
   \IF {$P_{i}^t \geq b$}
   \STATE Allocate the set of items $S_{i}^{t}$ to bidder $i$ with price $P_{i}^t$.
       \ELSE \RETURN
   \ENDIF
   \STATE Remove from $\mathcal{S}$ every bid made by bidder $i$ and every bid for a set of items that are not disjoint from $S_{i}^{t}$.
\ENDWHILE
\end{algorithmic}
\end{algorithm}

\begin{theorem}\label{thm:general CCA}
If bidders only bid on sets with cardinality at most $\C$, then the either the revenue of the CCA is at 
least $\frac{\Opt}{480\C^{2}\log n\log^2{m}}$ {or the social welfare of the CCA is at 
least $\frac{OPT}{24 \log n}$.}
 Thus, the welfare ratio is at most $O(\C^2 \cdot \log n \log^2 m)$.
\end{theorem}

Before proving Theorem~\ref{thm:general CCA}, let us generalize the time amplifying lemma for general bidders.
Let $X$ be the set of bidders that have been assigned items in the greedy algorithm and $\tilde{X}$ be the set of items that are allocated. Let $k\geq |X|$ be some integer. As in the unit-demand case, we will argue a generalization of Lemma~\ref{lem:time amplifying} for general bidders. With this Lemma, it is straightforward to argue that {either at least one bidder not in $X$ must make $b\cdot m$ bids on bundles of items which do not intersect with $\tilde{X}$, or many bidders still have high utility when the price ascending phase ends.}

\begin{lemma}[Time Amplifying for General Bidders]\label{lem:time amplifying general}
Let $S$ be a set of at least $4\C\cdot k$ bidders disjoint from $X$ ($|X|\leq k$), such that every bidder $i\in S$ has utility at least $u\geq 2\C\cdot b$ in round {$t\geq 5/2 \cdot b \cdot \C$}. If the greedy algorithm has revenue $g$ less than $k b-mn$, then in {any round up to} $2(t+1)-1$ (the mechanism can terminate before this), there is a subset of $S$ with at least $|S|-4\C\cdot k$ bidders such that each of them has utility at least $u-2\C\cdot b$.
\end{lemma}

\begin{proof}Let $t'$ be the first round that at least $4\C \cdot k$ bidders from $S$ have utility at most $u-2\C\cdot b$, and $S'$ be the set of these bidders. {It is not difficult to argue that if $t'$ does not exist, the conclusion holds.} 

Since every bidder $i\in S$ has utility $u$ in round $t$, by Fact~\ref{fact:monotone utility} and~\ref{fact:v geq u}, we know that for any set $S$ that $i$ bids on in the first $t$ rounds, $v_{i}(S)\geq u$. Let $M'$ be the subsets of items in $M-\tilde{X}$ that has ever been bid on by some bidder from $S'$ in the first $t$ rounds. The bidders of $S'$ totally make $|S'|\cdot t$ bids. How many of their bids can intersect with $\tilde{X}$? Let us assume the greedy algorithm allocates set $S_{i}^{r_{i}}$ to bidder $i$. Since bidders in $S$ are not selected by the greedy algorithm, they cannot bid on any item $j\in S_{i}^{r_{i}}$ after its price has reached $P_{i}^{r_{i}}$, so they can make at most {$P_{i}^{r_{i}}+n$} bids containing $j$. As the size of $S_{i}^{r_{i}}$ is at most $\C$, bidders from $S$ can make at most  {$\C P_{i}^{r_{i}}+\C n$} bids that intersect with $S_{i}^{r_{i}}$. Totally, bidders can make at most  {$\C g+\C nm \leq \C k b$} bids that intersect with $\tilde{X}$, as $\sum_{i\in X} P_{i}^{r_{i}}=g\leq k b- nm$. Therefore, bidders from $S'$ make at least  {$|S'|\cdot t-\C  k b$} bids only on subsets of $M'$. 

For any item $j\in M'$, there could be at most {$b+n$} {bids from bidders in $S$} containing $j$. As every bid is on a set with size at most $\C$, for any set that has ever been bid on, there can be at most  {$\C b+\C n$} other bids that intersect with it. Therefore, we can find at least {$\frac{|S'|\cdot t-\C k b}{\C b+\C n}$} disjoint bids on $M'$ that are made in the first $t$ rounds by bidders of $S'$. Let $T$ be the set of items bid by one of these disjoint bids. The total price for $T$ {at round $t'$} should be at least $2\C b$, otherwise at least one of the bidders in $S'$ will have utility greater than $u-2\C b$. 
On the other hand, for any item $j\in T$, bidders in $S'$ can bid on sets containing it only when its price is less than $b$ {or when the set intersects $\tilde{X}$, since otherwise the bid would have been selected by the greedy algorithm}. As $|T|\leq \C$, bids from $X$ and  bids intersecting $\tilde{X}$ must push up the price for $T$ by at least $\C\cdot (b-n)$. 
Since there are at least {
$\frac{|S'|\cdot t-\C k b}{\C b+\C n}$ such disjoint sets and at most $\C kb$ bids can intersect with $\tilde{X}$, the bids from $X$ on $M'$ must increase the total price of items in $M'$ by at least $\frac{b-n}{b+n}(|S'|\cdot t-\C k\cdot b) - \C^2 kb$. 
 Since $|X| \leq k$, we have 
}
 {
\begin{align*}
t' & \geq \frac{1}{\C k} \Big( \frac{b-n}{b+n}(|S'|\cdot t-\C k\cdot b) - \C^2 kb \Big)	 &&\\
   & \geq \frac{b-n}{b+n} \cdot (4 t -b)-b \cdot \C 						 && \hspace{-70pt}(|S'|\geq {4 \C \cdot k}) \\
   & \geq (1-\frac 2n)(4 t -b)-b \cdot \C 							 && \hspace{-70pt}(b\geq n^{2}) \\
&\geq  2(t+1) 									&& \hspace{-70pt}(t\geq \frac 52 b\cdot \C)
\end{align*}
}


By the definition of $t'$, it is straightforward to see in round $t'-1 \geq 2(t+1)-1$ there are at least $|S|-4\C\cdot k$ bidders from $S$ such that each of them has utility at least $u-2\C\cdot b$.
\end{proof}

We now have all the tools we need to prove Theorem~\ref{thm:general CCA}. \vspace{7pt}

\begin{prevproof}{Theorem}{thm:general CCA}
First let us fix the notations. In Lemma~\ref{lem:equal value group}, we make no assumption on the sets $\{R_i\}_{i\in[n]}$. In particular, we do not assume it is a singleton. Therefore, Lemma~\ref{lem:equal value group} also holds in the general case. Let $B$ and $v^{*}$ be the same as in Lemma~\ref{lem:equal value group}, and use $\mathcal{K}$ to denote the size of $B$. Let the threshold $b=\frac{v^{*}}{c_{1}\cdot\C\cdot\log m}$ and $k=\frac{\mathcal{K}}{c_{2}\cdot\C\cdot\log m}$ for some constants $c_{1}$ and $c_{2}$ that we will specify later. We prove by contradiction that the revenue of the greedy algorithm with our choice of $b$ is at least $g=\frac{\Opt}{6c_{1}c_{2}\cdot\C^{2}\cdot\log n \log^{2}m}$ {or the social welfare of the CCA} is at least $\frac{OPT}{24 \log n}$. 

As in Theorem~\ref{thm:matching CCA}, we can immediately show that $k\cdot b\geq g$. Thus, if the number of bidders selected by the greedy algorithm $|X|$ exceeds $k$, the revenue is clearly greater than $g$. So we can assume $|X|\leq k$. Let $X'$ be the set of bidders that are allocated a subset of items that  intersect with $\tilde{X}$ in $\vec{R}^{*}$, and $Y=B-X-X'$. Since $|X'|\leq |\tilde{X}|\leq \C|X|\leq\C k$, we have $|Y|\geq (c_{2}\cdot\C\cdot\log m-\C-1)k$. Take $t_{1}$ (if it exists) to be the first round that at least $4\C\cdot k$ bidders from $Y$, such that their allocated sets of items in $\vec{R}^{*}$ all have prices at least $2\C\cdot b$. 
We continue our proof with the following case analysis.

\begin{itemize}
\item {\bf Case (1):} $t_1$ does not exist. Then when the algorithm stops at round $t$, less than $4 \C \cdot k$ subsets of items allocated to bidders from $Y$ in $\vec{R}^{*}$ have prices at least $2 \C b$. Let $Y'$ be the bidders of $Y$ such that the total price of the set of items $R_i$ allocated to $i$ is at most $2 \C b$ at round $t$ (remind that $R_i$ refers to the set of items allocated to $i$ in $\vec{R}^*$). Notice that $|Y'| \geq |Y|-4 \C k\geq |Y|/2$ and $u_i^{t} \geq v^*-2\C b\geq v^{*}/2$ for any $i \in Y'$.
Fact~\ref{fact:disjoint bids} implies that the social welfare of the CCA is at least 
\[\frac{|Y| \cdot v^*}{4} \geq \frac{c_2 \C k \cdot \log m \cdot v^*}{8} \geq \frac{\mathcal{K} v^*}{8} \geq \frac{OPT}{24 \log n}.\]

\item {\bf Case (2):} $t_{1}$ exists. Let $T$ be the allocated set of items of one of the $4\C\cdot k$ bidder. As we have argued in Lemma~\ref{lem:time amplifying general}, {for any item $j$ in $T$}, only bidders from $X$ {or bids intersect with $\tilde{X}$} can bid on any set containing $j$ after its price reaches $b$. 
A simple calculation shows that bidders from $X$ must make at least $4 \C k \cdot (b-n)\cdot \C -\C^2 kb$ bids (as we notice in the proof of Lemma~\ref{lem:time amplifying general} that at most $\C^2 kb$ of total price increment is due to bids that intersect with $\tilde{X}$) in the first $t_{1}$ rounds. Hence, bidders from $X$ must make at least $3 (b-n) \C \cdot k -n$ bids, which means $t_{1}$ is at least $\frac{3(b-n)\cdot \C\cdot k-n}{|X|}$. Since $k\geq |X|$ and $b\geq n^{2}$, we have $t_{1}\geq \frac 52 \C\cdot b$.

Notice that in round $t_{1}-1$, there are at least $|Y|-ck$ bidders whose allocated items in $\vec{R}^{*}$ have prices less than $2\C\cdot b$. Let us call this set of bidders $Y_{1}$, and clearly they all have utility at least $u_{1}= v^{*}-2\C\cdot b$. Let $t_{2}$ (if it exists) be the first round that at least $4\C\cdot k$ bidders from $Y_{1}$ have utility no greater than $u_{2}=u_{1}-2\C\cdot b$. Let $Y_{2}$ be the set of bidders from $Y_{1}$ that still have utility at least $u_{2}$ in round $t_{2}-1$. Clearly, $|Y_{2}| \geq |Y_{1}|-4\C\cdot k$. Applying Lemma~\ref{lem:time amplifying general} on $Y_{1}$, $u_{1}$ and $t_{1}-1$ shows that $t_{2}-1\geq2 t_{1}-1$. Hence, $t_{2}\geq 2t_{1}$. 

There is nothing special about $Y_{1}$, $u_{1}$ and $t_{1}-1$. If we recursively define $t_{i}$ (if it exists) as the first round that at least $4\C\cdot k$ bidders from $Y_{i-1}$ have utility no greater than $u_{i}=u_{i-1}-2\C\cdot b$, and define $Y_{i}$ as the set of bidders from $Y_{i-1}$ that still have utility at least $u_{i}$ in round $t_{i}-1$, we can apply Lemma~\ref{lem:time amplifying general} on $Y_{i-1}$, $u_{i-1}$ and $t_{i-1}-1$ as long as they satisfy the conditions of Lemma~\ref{lem:time amplifying general} {and the CCA does not terminate before $t_{i}$ happens}. In that case, we have $t_{i}\geq 2 t_{i-1}$. How many times can we apply this Lemma before the conditions are violated?  Since every time the size of $Y_{i}$ decrease by at most $4\C\cdot k$ and $u_{i}$ decrease $2\C\cdot b$, we can apply Lemma~\ref{lem:time amplifying general} for at least ${\ell'}=\min\{[\frac{|Y_{1}|}{4\C\cdot k}], [\frac{u_{1}}{2\C\cdot b}] \}-1\geq \min\{[\frac{c_{2}-1}{4}\cdot \log m], [\frac{c_{1}-1}{2}\cdot \log m]\}$. If we let {$c_{1}=8$ and $c_{2}=10$, we have $\ell' \geq 2[\log m]+8$. Let $\ell = \log m+2$.} Remember that we might not be able to run the recursion till the conditions are violated, because the CCA might terminate before that. We now use case analysis to show that if we take $\ell=\log m+2$, then our claim holds no matter the CCA terminates before $t_{\ell}$ or not. 

\begin{itemize}
\item {\bf Case (i):} The CCA terminates between $t_{j}$ and $t_{j+1}$ where $j<\ell$. By our choice of the parameters, we have $|Y_i|- 4 \C k \geq |Y|/2$ and $u_i-2 \C b \geq v^*/2$. Thus if $t_{j+1}$ does not exist, Fact~\ref{fact:disjoint bids} implies that the welfare of the CCA is at least $(|Y_i|-4 \C k) \cdot (u_i-2 \C b) \geq \frac{|Y| \cdot v^*}{4} \geq \frac{OPT}{24 \log n}$.

\item {\bf Case (ii)} $t_\ell$ exists. It is straightforward to see that $t_{\ell}\geq 2^{\ell-1}\cdot t_{1}>2m\cdot t_{1}>(\C k+m)\cdot \C\cdot b$, which means there is a bidder $i$ in $Y_{\ell}$ that has made at least $(\C k+m)\cdot \C\cdot b$ bids in the price ascending phase. Since the revenue of the greedy allocation is less than $k\cdot b-nm$, so as we have argued in Lemma~\ref{lem:time amplifying general}, $i$ can make at most $\C^{2}\cdot k\cdot b$ bids on sets of items intersecting $\tilde{X}$. Therefore, $i$ makes at least $\C\cdot m\cdot b$ bids contained in $M-\tilde{X}$, then there must be a set $S\subseteq M-\tilde{X}$ that $i$ has bid on with price at least $b$. This is a contradiction, because $i$ would have been selected by the greedy algorithm.
\end{itemize}
\end{itemize}
\end{prevproof}

%

\section{Lower Bounds on the Welfare Ratio}\label{sec:dependence on C}

\subsection{Lower bounds for Unit-demand bidders}

In this section, we present lower bounds on the welfare ratio of the CCA. To begin, we show give a
polylogarithmic lower bound for unit-demand bidders.
\begin{theorem}\label{thm:noconstant1}
For any constants $\mathcal{N}$ and $\mathcal{M}$, there exists an instance with $n\geq \mathcal{N}$ 
unit-demand bidders and $m\geq \mathcal{M}$ items such that the social welfare of the CCA is 
only $\mathcal{O}\Big(\frac{\log\log m\cdot\Opt}{\log m}\Big)= \mathcal{O}\Big(\frac{\log n\cdot\Opt}{n}\Big)$. 
In particular, this implies that the welfare ratio of the CCA is at least $\Omega \Big( \frac{\log m}{\log \log m}\Big)$.

This implies that for any non-negative real numbers $a$ and $b$ such that $a+b<1$, the welfare ratio is at least $\Omega(n^{a}\cdot \log^{b} m)$.
\end{theorem}

Notice that Theorem \ref{thm:noconstant1} does not contradict  with the polylogarithmic upper bound 
obtained in Theorem~\ref{thm:matching CCA}. The number of items $m$ of our construction in Theorem \ref{thm:noconstant1} is exponential in the number of bidders $n$,
therefore $(\log n \cdot \log^2 m)$ is larger than $\frac{n}{\log n}$.
Theorem~\ref{thm:noconstant1} states that the welfare ratio is at least $\Omega(\frac{n}{\log n})$ if it only depends on $n$, and is at least $\Omega(\frac{\log m}{\log\log m})$ if it only depends on $m$.  
More generally, Theorem \ref{thm:noconstant1} implies that there is no welfare ratio dominating\footnote{If ratio $r$ has better dependence on both $n$ and $m$ than ratio $r'$, we say $r$ dominates $r'$.} any weighted geometric mean of the above two ratios, even if we allow dependency on both $n$ and $m$. On the one hand, this shows that if any approximation ratio has only sub-logarithmic dependence on $m$, it must have polynomial dependence on $n$. On the other hand, although Theorem~\ref{thm:matching CCA} has come quite close to the lower bound above, it has quadratic dependence on $\log m$, thus not contradicting with Theorem~\ref{thm:noconstant1}.

The idea behind the lower bound consists in constructing an instance that is essentially tight
with respect to the proof of Theorem~\ref{thm:matching CCA}. 
In particular, in the lower bound instance, we ensure that there are only two bidders who ever make large bids. 
All the other bidders only make small bids and their largest bids are on items they value 
the least! Thus, the revenue maximizing allocation selected by the CCA has very poor social welfare. 
Moreover, we have a similar hardness result for the general case.

\begin{proof}
Let $k,\ell$ be two integers.  We will show that for sufficiently large $k$ and $\ell$ we can construct an instance that has low social welfare as promised. To avoid cumbersome notations, we assume that the price increment of the CCA is $1$ when two bidders bid on it. Let us denote by $(u_n)_n$ the sequence such that $u_0=1$, $u_1=k+1$ and $u_p=k \cdot \Big( \sum_{i=0}^{p-1} u_i \Big)$ for any $p\geq 2$.

Let us consider the following instance of the CCA. There are $2k+2$ bidders $s_0,s_0',s_1,s_1'\ldots,s_k,s_k'$. For each bidder $s_i$ with $i \geq 1$ and for every $v\in[\ell]$, create a set of items $X_v^i$ such that there are $u_{\ell-v}$ items in $X_v^i$. In other words, bidder $s_i$ has one item of value $\ell$, $k+1$ items of value $\ell-1$, $u_2$ items of value $\ell-2$, etc.. Let $X_v=\cup_{i=1}^k X_v^i$ and we will refer to it as \emph{the set of items of value $v$}\footnote{In Section~\ref{sec:CCA}, we assume all bidders' values are multiples of some large integer $W$. We can easily modify our hard instances in this section by multiplying all values with $W$, and the same conclusion still holds.}. Moreover, we assume that bidder $s_i$ has the following preference rule over the items: if there are a few items with the same utility, the one with the lowest value is preferred. 
This preference rule can also be simulated by making a small modification to the valuation function of $s_i$. We will stick with the preferences rule, as we believe it makes the proof cleaner. The bidders $s_1',\ldots,s_k'$ are respectively copies of bidders $s_1,\ldots,s_k$. Since they have the same valuation functions and the same preference rules, we can assume that at any round they will bid on the same item.

Let us now describe the valuation function of $s_0$. For every $v \leq \ell$, the value of $s_0$ for any item in $X_v$ is $v$. Finally, we create a new item which is added in $X_\ell$ and we assume that $s_0$ has value $\ell$ for this item. We further assume that, this new item is the one preferred by $s_0$ among all items of value $\ell$. We also assume that $s_0$ has the following preference rule: if there are a few items with the same utility, the one with the highest value is preferred.  Note that the preference rule ensures that if an item of $X_v^i$ has price $p$ and an item of $X_{v-1}^i$ has price $(p-1)$, bidder $s_i$ bids on the item in $X_{v-1}^i$ while $s_0$ bids on the item in $X_v^i$. Again we create a bidder $s_0'$ as a copy of $s_0$.

The fact that bidder $s_i$ prefers the item with the lowest value (and the fact that $X_{v-1}^i$ is large compared to $X_v^i$) will ensure that the welfare of the allocation of the CCA is small compared to the optimal one. The key step is to prove the following statement:

\begin{claim}\label{clm:nopositivebid}
None of the bidders $s_1,s_1',\ldots,s_k,s_k'$ makes a bid of price at least $2$. Moreover, all their bids of price $1$ are performed on items of $X_1$.
\end{claim}
\begin{proof}
At round $t=0$, all the bidders $s_1,s_1'\ldots,s_k,s_k'$ bid on their favorite items (the items of value $\ell$ for their own valuation functions). Both bidders $s_0$ and $s_0'$ bid on their favorite item, which is the special item in $X_\ell$. So the prices of all these items increase by one (by definition of the price increment). Now, all items of $X_\ell$ have price $1$ and all the other items have price $0$.

By construction of our preference rules, bidders $s_1,s_1'\ldots,s_k,s_k'$ start bidding on items of value $\ell-1$ and $s_0,s_0'$ bid on items of value $\ell$. Recall that since $s_i,s_i'$ are copies of the same bidder, they bid on the same item at any round. Since, for every $i \leq k$, $X_{\ell-1}^i$ contains $k+1$ items, bidder $s_i$ and $s_i'$ need to spend $k+1$ rounds to increase the prices of all the items in $X_{\ell-1}^i$ from $0$ to $1$. Since there are $k+1$ items in $X_\ell$, bidders $s_0,s_0'$ also need $k+1$ rounds to increase the prices of all the items of value $\ell$ from $1$ to $2$. So by the end of round $t=k+2$, all items in $X_\ell$ have price $2$ and all items in $X_{\ell-1}$ have price $1$. All the other items have price $0$.

We now repeat this argument by induction. More precisely, we will show that for every $v$, there exists a round $t_v$ such that:
\begin{itemize}
 \item All items in $X_{v-p}$ have price $0$ by the end of round $t_v$ for any $p>0$.
 \item All items in $X_{v+p}$ by the end of round $t_v$ have price $p+1$ for any $p \geq 0$.
\end{itemize}
We have already showed that $t_{\ell} = 1$ and $t_{\ell-1}=k+2$. Let us prove that the existence of $t_v$ implies the existence of $t_{v-1}$. Let us study the structure of the bids at round $t_v$. By construction of the preference rule, bidders $s_1,s_1',\ldots,s_k,s_k'$ prefer bidding on items of value $(v-1)$ at price $0$. On the other hand, the bidders $s_0,s_0'$ prefer bidding on items of value between $v$ and $\ell$ (more precisely they prefer items in $X_\ell$, then $X_{\ell-1}$ and so on). Since by definition of the sequence $(u_n)_n$, the number of items in $X_{v-1}^i$ is is equal to the number of items in $\cup_{v'=v}^\ell X_{v'}$. So, for every $i$, the number of rounds needed by the bidders $s_i$ and $s_i'$ to increase the prices of all the items of $X_{v-1}^i$ from $0$ to $1$ is exactly the number of rounds needed by $s_0$ and $s_0'$ to increase all the prices of all the items in $X_{v'}$ for any $v \leq v' \leq \ell$ by one. 

Consider finally by the end of round $t_1$. All items in $X_v$ have price $v$ for any $v \leq \ell$. The preference rule ensures that, bidders $s_1,s_1'\ldots,s_k,s_k'$ will increase the prices of items of $X_1$\footnote{We assume bidders still participate in the mechanism even if they have utility $0$. An alternative approach to force this behavior is to add a tiny $\epsilon$ to all the values.}. On the other hand, bidders $s_0$ and $s_0'$ increase the prices of all the items in $X_v$ for $v\geq 2$ before increasing the prices of items in $X_1$. Since the size of $X_1^i$ is precisely the size of $\cup_{v=2}^\ell X_v$ for every $i$, after $|X_1^i|$ rounds the price for each item has increased by one. Now, all items in $X_v$ have price $v+1$ for every $v \leq \ell$. As a result, all bidders have negative utility and they drop out.
\end{proof}

 Finally, we compare the optimal welfare with the welfare of the allocation of the CCA. The optimal solution allocates to each of $s_{0},s_1,\ldots,s_k$ her favorite item of value $\ell$. Then we can allocate to $s_{0}', s_1',\ldots,s_k'$ an item in $X_{\ell-1}^i$ for every $i$. The welfare of the optimal allocation is $\Theta(k\ell)$. The allocation of the CCA allocates an item of $X_\ell$ at price $\ell$ to both $s_0$ and $s_0'$. All the other bidders only bid on items at price $0$ except for items in $X_{1}$, on which they make bids at price $1$, so the CCA allocates an item of value $1$ to each of $s_1,s_1',\ldots,s_k,s_k'$. Therefore, the welfare of the CCA is $2\ell+2k$. Note that the number of items in this construction is $m=\Theta(k^\ell)$ (the number of items of value $v-1$ is essentially $k$ times the number of items of value $v$). Now if we choose $k=\Theta(\log m)$, and the number of bidders is $n=\Theta(k)$. If we let $k=\Theta(m)$, then $\ell=\Theta\Big(\frac{\log m}{\log\log m}\Big)$ and $n=\Theta(\log m)$. The welfare of the CCA is $\mathcal{O}\Big(\frac{\log\log m\cdot\Opt}{\log m}\Big)=\mathcal{O}\Big(\frac{\log n\cdot\Opt}{n}\Big)$. 
 
 If we can guarantee the welfare of the CCA is at least $\Omega\Big(\frac{\Opt}{n^{a}\cdot \log^{b}m}\Big)$ for some $a$ and $b$ such that $a+b<1$, then clearly ${\Omega}\Big(\frac{\Opt}{n^{a}\cdot \log^{b}m}\Big) \geq \min\Big\{{\Omega}\Big(\frac{\Opt}{\log^{(a+b)} m}\Big), {\Omega}\Big(\frac{\Opt}{n^{(a+b)}}\Big)\Big\}$. Since $a+b<1$, this guarantee contradicts with our instance above.
\end{proof}

\subsection{Lower bounds for General Bidders}

Let us now prove that a similar lower bound can be obtained for general bidders.

\begin{theorem}\label{thm:general upperbound}
For any constant $\mathcal{N}$ and $\mathcal{M}$, there exists an instance with $m\geq \mathcal{M}$ items 
and $n\geq \mathcal{N}$ bidders who only bid on bundles of cardinality at most $\mathcal{C}\leq m^{c}$, for 
some absolute constant $c\in(0,\frac 12)$, such that the CCA's social welfare is 
only $\mathcal{O}\Big(\frac{\log\log m\cdot \Opt}{\mathcal{C}\cdot \log m}\Big)=\mathcal{O}\Big(\frac{\log n\cdot\Opt}{n}\Big)$. 
In particular, this implies that the welfare ratio of the CCA is at least $\Omega\Big(\C \cdot \frac{\log m}{\log \log m}\Big)$.

This implies that for any non-negative real numbers $a$ and $b$ such that $a+b<1$, the welfare ratio is at least $O(n^{a}\cdot (\C \cdot \log m)^b)$.
\end{theorem}

\begin{proof}
  Let $k,\ell, \mathcal{C}$ be three integers. We fix $\cal C$ to be an odd constant, and we will show that for sufficiently large $k$ and $\ell$ we can construct an instance that has low social welfare as promised. The proof uses the same ingredient as the proof of Theorem~\ref{thm:noconstant1} but is slightly more involved.
 We assume that the price increment for any item is $1/2$ times the number of bidders bidding on it. Let $u_n$ be the following sequence: $u_0=1$, $u_1=k \cdot \frac{\mathcal{C}-1}{2}+1$ and $u_n = k \cdot \frac{\mathcal{C}-1}{2} \cdot \sum_{i=0}^{n-1}u_{i}+1$. Note that since $\mathcal{C}$ is odd, $u_n$ is a sequence of integers.
 
The instance of the CCA has $k \cdot \mathcal{C}+2$ bidders. They are denoted by :$s_1^i,\ldots,s_{\mathcal{C}}^i$ for every $i \leq k$ and two special bidders denoted by $s_0$ and $s_0'$. Let $B_{i}$ be the set of bidders $s_1^i,\ldots,s_{\mathcal{C}}^i$. Before describing the construction, we first introduce the $v$-gadgets. 

\paragraph*{$v$-gadgets.}
We say a set $K$ of $\C\choose 2$ items form a $v$-gadget for a set of bidders $\{s_{1},\ldots, s_{\C}\}$, if we can use edges of a clique on $\cal C$ vertices to encode $K$, such that for any bidder $s_{i}$ only the subset of items corresponding to all edges incident to the vertex $i$ has value $\C \cdot v$ to her, while all the other subsets have value $0$. Note that, for any item of the $v$-gadget, there are exactly two bidders containing it in their interested bundles (since each edge has two endpoints). 

\paragraph{Instance.}
Before introducing formally the instance, let us briefly describe it. As in the matching case, all the bidders only bid on items with price $0$ till the last round except bidders $s_0$ and $s_0'$. At any round, bidders in $B_{i}$ for every $i$ bid on a disjoint $v$-gadget. We will create the appropriate number of $v$-gadgets (and an appropriate preference rule) to ensure that none of the bids is positive except on items with low value. Therefore, we can argue the revenue optimal allocation selected by the CCA has low social welfare.

For every $v\in[\ell]$ and every $i \in[k]$, we create a set $X^{i}_{v}$ of $u_{\ell-v}$ copies of the $v$-gadgets for bidders in $B_{i}$. Additionally, for every bidder $s_{j}^{i}$, we create a special item $r^{i}_{j}$.  Each bidder $s_{i}^{j}$ has value $1/2$ for her special item and $\C\cdot v$ for any bundle of items corresponding to the set of all edges of the $i$-th vertex in some $v$-gadget of $X^{i}_{v}$ for all $v\in[\ell]$, but has value $0$ for any other bundle. We use $X_{v}$ to denote the union of all $X_{v}^{i}$ and refer to it as the set of items with value $v$. 

To simplify the proof, we assume that all bidders in $B_{i}$ obey the following preference rule: (i) all bidders have the same preference order over the $v$-gadgets of the same utility, \emph{i.e.} if two $v$-gadgets have exactly the same prices\footnote{Each item and its counterpart in the other $v$-gadget shares the same price.}, then all bidders prefer the same $v$-gadget; (ii) each bidder $s^{i}_{j}$ prefers to bid on the set of items with lower value when there is a tie; (iii) a bidder will not bid on any $v$-gadget with utility $0$, but will still bid on her special item at price $\frac 12$. As we mentioned in the proof of Theorem~\ref{thm:noconstant1}, the preference rule above can be easily simulated with slight perturbations on the valuation functions.

Now we describe the valuation function of $s_0$.  We first create a special item $r_{0}$ with value $1/2$, and add a set of $\mathcal{C}$ new items to $X_{\ell}$ such that the whole set has value $\C\cdot v$. For any $v$-gadget, we partition it into $\frac{\mathcal{C}-1}{2}$ disjoint subsets of size $\mathcal{C}$, such that $v_{0}$'s value for each of these subsets is $\mathcal{C} \cdot v$. Except the bundles mentioned above, all other bundles have value $0$. We also assume that $s_{0}$ obeys a preference rule: (i) among all bundles of value $\mathcal{C} \cdot \ell$, the bundle of $\mathcal{C}$ new items is preferred; (ii) $s_{0}$ prefers the set of items with higher value when there is a tie. Finally, we create $s_{0}'$ as a copy of bidder $s_0$, such that she has the same valuation function and same preference over bundles.

Our preference rule ensures that the CCA's welfare is small. The key step of our proof is to show the following statement:

\begin{claim}
None of the bidders $s_1^j,\ldots,s_{\mathcal{C}}^j$ for any $j \leq k$ makes positive bids except on her special item on which she makes a bid at price $\frac 12$.
\end{claim}
\begin{proof}
The proof is similar to Claim~\ref{clm:nopositivebid}.

At round $t=0$, all bidders $s_1^j,\ldots,s_{\mathcal{C}}^j$ bid on their favorite bundles (the unique bundle of value $\mathcal{C} \cdot \ell$ in the $\ell$-gadget). Moreover $s_0$ and $s_0'$ bid on the bundle of new items in $X_{\ell}$ (the set of $\C$ items we create when we define the valuation function of $s_0$). So by the end of the first round, the prices of all items of $X_\ell$ are increased to $1$, while the prices for all the other items remain at $0$.

By definition of our preference rules, bidders $s_1^i,\ldots,s_{\mathcal{C}}^i$ start bidding on $(\ell-1)$-gadgets and $s_0,s_0'$ continue to bid on the $\ell$-gadgets. Since there are $k \cdot \frac{\mathcal{C}-1}{2}+1$ different $(\ell-1)$-gadgets in $X_{\ell-1}^{i}$ for every $i \in [k]$, $s_1^i,\ldots,s_{\mathcal{C}}^i$ need $k \cdot \frac{\mathcal{C}-1}{2}+1$ rounds to increase the prices of all the items in 
$X_{\ell-1}$ from $0$ to $1$. Since there are $k \cdot {\mathcal{C} \choose 2} +\mathcal{C}$ items in $X_{\ell}$ and these items are partitioned into subsets of size $\mathcal{C}$, bidders $s_0,s_0'$ also need $k \cdot \frac{\mathcal{C}-1}{2}+1$ rounds to increase the prices of all these items of value $\ell$ from $1$ to $2$. So by the end of round $t=k \cdot \frac{\mathcal{C}-1}{2}+2$, all the items in $X_\ell$ have price $2$ and all the items in $X_{\ell-1}$ have price $1$. All the other items still have price $0$. 

We now repeat this argument by induction. More precisely, let us prove that, for every $v$, there exists a round $t_v$ such that:
\begin{itemize}
 \item All items in $X_{v-p}$ have price $0$ by the end of round $t_v$ for any $p>0$.
 \item All items in $X_{v+p}$ by the end of round $t_v$ have price $p+1$ for any $p \geq 0$.
\end{itemize}
We have already showed that $t_{\ell}=1$ and $t_{\ell-1}=k \cdot \frac{\mathcal{C}-1}{2}+2$. Now we will prove that the existence of $t_v$ implies the existence of $t_{v-1}$. Let us first understand the structure of the bids at round $t_v$. By construction of the preference rule, $s_1^i,\ldots,s_{\mathcal{C}}^i$ prefer bidding on bundles of the $(v-1)$-gadgets in $X_{v-1}^{i}$ which have price $0$ for every $i\in[k]$. Since they have the same preference over the gadgets, they bid on the same $(v-1)$-gadget.
In the meantime, the bidders $s_0,s_0'$ are bidding on $v'$-gadgets with $v \leq v' \leq \ell$ (more precisely, they first bid on $\ell$-gadgets, then on $(\ell-1)$-gadgets, etc.). By the definition of the sequence $(u_n)_n$, the number of $(v-1)$-gadgets in $X^{i}_{v-1}$ equals to $\frac{ | \cup_{v'=v}^\ell X_{v'}  |}{\mathcal{C}}$. Thus, the number of rounds needed by the bidders $s_1^i,\ldots,s_{\mathcal{C}}^i$ to increase the prices of all the items in $X_{v-1}^{i}$ from $0$ to $1$ is exactly the number of rounds needed by $s_0$ and $s_0'$ to increase the prices of all the items in $X_{v'}$ for any $v \leq v' \leq \ell$ by $1$. As a result, $t_{v-1}$ exists and equals to $t_{v}+|X^{i}_{v-1}|$. 

By the end of round $t_1$, every bidder has utility $0$ on any item in $X_v$ for any $v \in [\ell]$. 
In the next round, all bidders bid on their special items.
By the end of round $t_{1}+1$, each items $r^{i}_{j}$ has price $\frac12$ and $r_{0}$ has price $1$. Thus, bidders $s_0$ and $s_0'$ drop out (since the last item reach price $1$) and the other bidders will bid on their special items for another round at price $1/2$. By the end of round $t_{1}+2$, the mechanism terminates as the stopping condition is met. Namely, all bids are disjoint and the revenue optimal allocation does not conflict with the final allocation.\end{proof}
Now we are ready to compare the welfare of the CCA to the optimal social welfare. 
Consider the following allocation (which is not optimal but sufficient to show our result): for every $i$, assign $s_1^i$ its bundle in the $\ell$-gadget, each of the other bidders her bundle in a distinct $(\ell-1)$-gadget (assuming $k\geq 2$). The welfare of this allocation is larger than $\Theta(k \cdot \mathcal{C}^2 \cdot \ell)$. 
In the revenue optimal allocation selected by the CCA,  two disjoints sets of value $\mathcal{C} \cdot \ell$ are allocated to $s_{0}$ and $s_{0}'$, and all the other bidders only receive their special items of value $\frac 12$. The welfare of the CCA is only $2 \mathcal{C} \cdot \ell + k \cdot \mathcal{C}/2$. 
Note that the number of items in this construction is $\Theta((k \cdot \mathcal{C})^\ell)$ (the number of items of value $v-1$ is essentially $k \cdot \mathcal{C}$ times the number of items of value $v$) and the number of bidders is $\Theta(\C\cdot k)$.
Since $\log\C \leq c\cdot \log m$ for some absolute constant $c$, we choose $k$ to be $\Theta(\log m)$, then $\ell=\Theta\Big(\frac{\log m}{\log \C+\log\log m}\Big)$ and $n=\Theta(\C\cdot \log m)$. The welfare of the CCA  is $\mathcal{O} \Big( \frac{(\log\log m+\log \C)\cdot \Opt}{\mathcal{C} \cdot \log m} \Big)=\mathcal{O}\Big(\frac{\log n\cdot\Opt}{n}\Big)$. For any fixed integer $\C$, this is the same as $\mathcal{O} \Big( \frac{\log\log m\cdot \Opt}{\mathcal{C} \cdot \log m} \Big)$ 

If we can guarantee the welfare of the CCA is at least $\Omega \big(\frac{\Opt}{n^{a}\cdot (\C \cdot \log m)^b}\big)$ for some $a$ and $b$ such that $a+b<1$, then clearly $\Omega \big(\frac{\Opt}{n^{a}\cdot (\C \cdot \log m)^b}\big) \geq \min\Big\{{\Omega}\Big(\frac{\Opt}{(\C \log m)^{(a+b)}}\Big), {\Omega}\Big(\frac{\Opt}{n^{(a+b)}}\Big)\Big\}$. Since $a+b<1$, this guarantee contradicts with our instance above.
\end{proof}

\section{The Importance of the Price Increment and the Stopping Rule}\label{app:examples}

In this section, we show that if price increments do not depend of the demand or if the stopping condition is replaced by the one of the SMRA (the auction stops when the bids are disjoint), then we can obtain lower bounds on the welfare ratio.

\subsection{Welfare Under the Original CCA}\label{sec:example 1 increment}
In the original CCA model designed by Porter et al.~\cite{PRR03}, the price increment is always $1$ when there is an excess demand. In this section, we show that the welfare under this model could be as low as $O(\frac{\Opt}{\poly(n)})$ even for simple valuations such as additive and unit-demand valuations. The valuation function of bidder $i$ is \emph{additive} if a bidder's  value for set $S$ is the sum of the values of the items in $S$.

\begin{lemma}
Assume that the price increment does not depend on the demand.
For any constants $\mathcal{N}$ and $\mathcal{M}$, there exists an instance with $n\geq \mathcal{N}$ 
unit-demand bidders and $m\geq \mathcal{M}$ items such that the social welfare of the CCA is 
only $O(\frac{\Opt}{\sqrt n})=O(\frac{\Opt}{\sqrt{m}})$.

This implies that for any non-negative real numbers $a$ and $b$ such that $a+b \leq 1$, the welfare ratio is at least $\Omega(n^{a/2} \cdot m^{b/2})$.
\end{lemma}

\begin{proof}
Let $n \geq \max(\mathcal{N},\mathcal{M})$.
Consider a unit-demand auction with items $\{0,1,2,\dots, n\}$. Let $t$ be a real defined below.
We have two classes of bidders. In the first class we have $n$ bidders and
for any package $S$, bidder $i$ ($1\le i\le n$) has a value 
$$v_i(S)=\max_{j\in S} v_i(\{j\})$$
where
\[
 v_i(\{j\}) =
  \begin{cases}
     t \cdot V & \text{if } j=0, \\
   V & \text{if } j=i \\
   0       & \text{otherwise} 
  \end{cases}
\]
In the second class we have $2\sqrt{n}$ bidders. For each $0\le \ell \le \sqrt{n}-1$ we have
two identical bidders. The two identical bidders $\ell$ have a valuation function 
$$v_{\ell}(S)=\max_{j\in S} v_{\ell}(\{j\})$$
for any package $S$, where
\[
 v_{\ell}(\{j\}) =
  \begin{cases}
   V & \text{if } j\in H_{\ell}=\{\ell\cdot \sqrt{n}+1, \ell\cdot \sqrt{n}+2, \dots, (\ell+1)\cdot \sqrt{n}\} \\
   0       & \text{otherwise} 
  \end{cases}
\]
Let us examine what happens in this auction. Each pair of bidders in the second class will 
bid on the cheapest item in the set $H_{\ell}$ of cardinality $\sqrt{n}$. In the case of a tie, we may assume
they both bid on the smallest index item.\footnote{This can be enforced by adding a small
perturbation to the valuation function.}
It follows that after $V\cdot \sqrt{n}$ rounds the price of every item in $H_{\ell}$ will reach $V$
and then the pair of bidders $\ell$ will drop out of the auction.

Meanwhile, bidder $i$ in the first class will continually bid on item $0$, at least until its price is $t \cdot V-V$.
This will take $(t-1)\cdot V$ rounds. But at this time, provided we set $t\ge\sqrt{n}+1$, the price of item
$i$ will be above $V$. So bidder $i$ will continue to bid on item $0$ until its price reaches $t\cdot V$ when
it will drop out of the auction.

Given this set of bids what is the optimal allocation? For $0\le \ell \le \sqrt{n}-1$, both bidders of type $\ell$ 
in the second class  will be allocated one item from the set $H_{\ell}$. 
Furthermore, exactly one bidder $i$ of the first class, where $1\le i\le n$,
will be allocated the item $0$. The other bidders of the first class do not make any bid on any other items
and thus are not allocated anything. Thus the total welfare of this allocation is
$$2\sqrt{n}\cdot V + t\cdot V = (3\sqrt{n}+1)\cdot V$$
On the other hand the optimal allocation has value
$$n\cdot V + t\cdot V = (n+\sqrt{n}+1)\cdot V$$
Hence the welfare ratio is at least $\frac13\sqrt{n}$.
\end{proof}

\begin{lemma}
Assume that the price increment does not depend on the demand.
For any constants $\mathcal{N}$ and $\mathcal{M}$, there exists an instance with $n\geq \mathcal{N}$ 
bidders that bid on sets of size at most $2$ and $m\geq \mathcal{M}$ items such that the social welfare of the CCA is 
only $O(\frac{\Opt}{n})=O(\frac{\Opt}{m})$.

This implies that for any non-negative real numbers $a$ and $b$ such that $a+b \leq 1$, the welfare ratio is at least $\Omega(n^a \cdot m^b)$.
\end{lemma}
\begin{proof}
Let $n \geq \max(\mathcal{N},\mathcal{M})$.
Consider an auction with items $\{0,1,\ldots, n\}$. Let there be $2n$ bidders.
In particular, for each $1\le i\le n$ there are two identical bidders of Type $i$.
Each Type $i$ bidders has a $2$-demand function ${\bf v}_i$ , that is, additive up to $2$ items. Each individual
item is valued as:
\[
 v_i(\{j\}) =
  \begin{cases}
   V & \text{if } j\in \{0,i\} \\
   0       & \text{otherwise} 
  \end{cases}
\]
Here $V$ is a large integer. Now let's run the CCA with a price increment of $\epsilon=1$ 
for excess demand items. We begin with all price equal to zero. Then each Type $i$ bidder bids for
the pair $\{0,i\}$ as it provides a utility of {$2V$}. Not that since there are two bidders for each type
the every item is in excess demand, so all prices rise to $1$. Each Type $i$ bidder then still bids for
the pair $\{0,i\}$, and continues to do so until every item has a price of $V$. The auction then terminates
with all prices equal to $V+1$. But now, given these bids, when we find a maximum revenue allocation
we may only accept one bid since every submitted bid contained the item $0$. Thus we obtain a
welfare of $2V$. In contrast, it is easy to see that the optimal welfare is $(n+1)\cdot V$. Hence the
welfare ratio is $\frac12(n+1)$.
 \end{proof}
Finally note that in~\cite{BSZ13}, a similar example shows the poor performance of the
Porter et al. mechanism.

\subsection{Bad Examples using the SMRA Stopping Condition}\label{sec:stopping}
In our CCA and as well as the original model of Porter et.al~\cite{PRR03}, the CCA terminates when there is no excess demand induced by the bids in the current round
{\em and} the maximum revenue allocation over all rounds is not in conflict with the current round bids.
A simpler condition is to stop immediately once the bids become disjoint as in the SMRA. This simpler rule is now used in recent real-world CCA auctions. In the following example, we show that the SMRA stopping condition has arbitrarily bad performance when used in the CCA. \\

\begin{lemma}
 There does not exist any garantee on the welfare ratio if the CCA stops when the bids are disjoint.
\end{lemma}
\begin{proof}
Consider an auction with $3$ additive bidders and $4$ items. Let $c$ be an integer larger than $2$. Here are the bidders' valuations
\[
 v_1(\{j\}) =
  \begin{cases}
   V & \text{if } j=1 \\
   2V & \text{if } j=2\\
   0 & \text{otherwise} 
  \end{cases}
\]

\[
 v_2(\{j\}) =
  \begin{cases}
   V & \text{if } j=4 \\
   2V & \text{if } j=3\\
   0 & \text{otherwise} 
  \end{cases}
\]

\[
 v_3(\{j\}) =
  \begin{cases}
   c\cdot V & \text{if } j \in \{2,3\} \\
   0 & \text{otherwise} 
  \end{cases}
\]
At round $0$, every bidder bids on their favorite set. Bidder $1$ bids on set $\{1,2\}$, bidder $2$ bids on set $\{3,4\}$ and bidder $3$ bids on set $\{2,3\}$. It is easy to verify that at round $t\leq V$, these are still the sets bidders bid on. Both item $1$ and $4$ have price $t$, and both item $2$ and $3$ have price $2t$. At round $V+1$, bidder $1$ and $2$ drop out as they have utility $0$, while bidder $3$ still bids on set $\{2,3\}$ with total price $4V+4$. If we stop at this round, the revenue optimal allocation gives item $1$ and $2$ to bidder $1$ with price $3V$ and item $3$ and $4$ to bidder $2$ with price $3V$. The welfare for this allocation is $6V$. 

Since $c > 2$, the welfare maximizing allocation should give bidder $3$ both items $2$ and $3$, give bidder $1$ item $1$ and give bidder $2$ item $4$. The welfare of this allocation is $2(c+1)\cdot V$. So the welfare ratio of the CCA with this stopping rule is at least $\frac{c+1}{6}=\Theta(c)$. Since $c$ can be arbitrarily large and does not depend on $n$ and $m$, the welfare ratio can be arbitrarily bad which achieves the proof.
\end{proof}


\end{document}